\documentclass[twocolumn]{autart}    

\usepackage{graphics}
\usepackage{graphicx}
\usepackage{amsmath}
\usepackage{amssymb}
\usepackage{epsfig,wrapfig}
\usepackage{amsmath}
\usepackage{color}
\usepackage{float}
\usepackage{amsmath}
\usepackage{amsfonts}
\usepackage{amssymb}
\usepackage{calrsfs}
\usepackage{cite}
\usepackage{verbatim}
\usepackage[title]{appendix}
\usepackage{xcolor}

\usepackage{wrapfig}

\newcommand{\Eta}{\Tilde{\eta}}
\newcommand{\R}{\mathbb{R}}
\newcommand{\BM}{\begin{bmatrix}}
\newcommand{\EM}{\end{bmatrix}}

\newcommand{\be}{\begin{equation}\begin{aligned}}
\newcommand{\ee}{\end{aligned}\end{equation}}

\newtheorem{definition}{Definition}
\newtheorem{theorem}{Theorem}

\newtheorem{corollary}{Corollary}

\newtheorem{remark}{Remark}
\newtheorem{assumption}{Assumption}
\newtheorem{example}{Example}

\begin{document}

\begin{frontmatter}

\title{Learning Invariant Subspaces of Koopman Operators--Part 1: A Methodology for Demonstrating a Dictionary's Approximate Subspace Invariance} 

\thanks[footnoteinfo]{ Corresponding Author at: The Biological Control Laboratory at UC Santa Barbara, Santa Barbara, California 93106}

\author[ey]{Charles A. Johnson}\ead{cajohnson@ucsb.edu}$^{,*}$,    
\author[sh]{Shara Balakrishnan}\ead{sbalakrishnan@ucsb.edu}, 
\author[ey]{Enoch Yeung}\ead{eyeung@ucsb.edu}               

\address[ey]{Department of Mechanical Engineering, University of California, Santa Barbara, 93106, United States}     

\address[sh]{Department of Electrical and Computer Engineering, University of California, Santa Barbara, 93106, United States}


\begin{abstract}                         
Koopman operators model nonlinear dynamics as a linear dynamic system acting on a nonlinear function as the state. 
This nonstandard state is often called a Koopman observable and is usually approximated numerically by a superposition of functions drawn from a \textit{dictionary}. In a widely used algorithm, \textit{Extended Dynamic Mode Decomposition}, the dictionary functions are drawn from a fixed class of functions.  Recently, deep learning combined with EDMD has been used to learn novel dictionary functions in an algorithm called deep dynamic mode decomposition (deepDMD). The learned representation both (1) accurately models and (2) scales well with the dimension of the original nonlinear system. In this paper we analyze the learned dictionaries from deepDMD and explore the theoretical basis for their strong performance. We explore State-Inclusive Logistic Lifting (SILL) dictionary functions to approximate Koopman observables.  Error analysis of these dictionary functions show they satisfy a property of subspace approximation, which we define as uniform finite approximate closure. Our results provide a hypothesis to explain the success of deep neural networks in learning numerical approximations to Koopman operators. Part 2 of this paper, \cite{johnson2022augSILL}, will extend this explanation by demonstrating the subspace invariant of heterogeneous dictionaries and presenting a head-to-head numerical comparison of deepDMD and low-parameter heterogeneous dictionary learning.
\end{abstract}

\end{frontmatter}

\section{Introduction}



Koopman operator theory considers an alternate representation of a dynamic system where the state evolution of a nonlinear system is linear.  In this representation, the concepts of vibrational, growth, and decay modes in linear systems can be directly extended to nonlinear systems \cite{mezic2005spectral}. These modes address problems in the field of fluid mechanics \cite{rowley2009spectral, schmid2010dynamic, mezic2013analysis} and disease modeling \cite{proctor2018generalizing}, attack identification in the power grid \cite{nandanoori2020model}, as well as programming the steady state of biological systems \cite{hasnaindata} or extracting new biosensors \cite{hasnain2022learning}.  The spectral properties of the linear Koopman operator in this function space connects to classical methods for model reduction, validation, identification and control \cite{mauroy2020Koopman}. Since Mezić's first paper on the spectral properties of the Koopman operator, computing Koopman modes has become a major research focus \cite{budivsic2012applied, williams2014kernel}. The central algorithm for computing these modes is dynamic mode decomposition (DMD). Extensions of DMD, such as extended dynamic mode decomposition (EDMD) \cite{williams2015data} enable the expression of strong nonlinearities in the model. Others, such as the robust approximation algorithm presented in \cite{sinha2018robust}, enable high fidelity modeling in the presence of both process and measurement noise.

DMD-based methods, such as EDMD are a useful set of tools to model nonlinear dynamics from data.  In EDMD, one chooses a nonlinear function space implicitly defined by the span of a predefined, homogeneous function dictionary and then, individual modes are learned to compute a low-rank approximation to a linear operator that evolves the function space forward in time. Because a human typically specifies the dictionary functions in EDMD, the resulting models tend to be high dimensional. One approach to learning lower dimensional models is through the SINDy algorithm. This algorithm uses sparse regression to project to a lower dimensional subspace of the nonlinear function space initially chosen \cite{brunton2016discovering}. Another, Symmetric Subspace Decomposition, approximates the maximal Koopman-invariant subspace in an iterative fashion \cite{haseli2021learning}. 

A second approach we have explored is to use deep learning to learn a function dictionary or a set of observables during training \cite{yeung2019learning}. The deepDMD algorithm uses a form of stochastic gradient descent (SGD) to train a deep artificial neural network to select observables.  Because the neural network is a universal function approximator \cite{hornik1991approximation}, it can learn a large class of continuous observables on a compact set and parameterize the approximate Koopman operator. The deepDMD algorithm has been generally successful at learning both small and large scale nonlinear systems \cite{yeung2019learning}.

Other approaches leverage deep learning to build efficiently parameterized Koopman models with high fidelity.  Methods using linear recurrent autoencoders build effective, low-dimensional Koopman models that are not measurement-inclusive \cite{wehmeyer2018time, lusch2018deep, otto2019linearly}. To do this, they pass observables from the learned function space, via the decoder, back to the original model's state-space. Other work has, under the assumption of ergodicity, delegated the process of choosing a dictionary for EDMD to deep learning \cite{takeishi2017learning}. 


Neural networks are typically black box models. We use dictionaries of functions inspired from deep learning to build Koopman models which both capture the system dynamics and approximate the dynamics of the nonlinear dictionary functions. The success of these models provides insight into why deep learning models have the capacity to approximate Koopman operators well. We explain this success by demonstrating that these dictionaries satisfy uniform finite approximate closure, a property tied to the subspace invariance of these dictionaries.

In Section \ref{sec:KLP}, we introduce the Koopman generator learning problem and define terms relevant to the subspace invariance of a Koopman dictionary. In Section \ref{sec:AppendixNotation}, we introduce notation. In Section \ref{sec:SILL}, we introduce a dictionary inspired by the dictionary functions of deepDMD, the SILL dictionary. In Sections \ref{sec:SILLclosure} and \ref{sec:averageErrorSILL}, we demonstrate the SILL dictionary's approximate subspace invariance. We draw our conclusions in Section \ref{sec:conclusion}.

\section{The Koopman Generator Learning Problem}\label{sec:KLP}
In this section we introduce the problem of choosing a function space and approximate Koopman generator to model a dynamic system.  This problem is not convex, but, when approximated well, it gives an accurate, data-driven, linear model of a nonlinear system. 

Consider a nonlinear, time-invariant, autonomous system with dynamics
\begin{equation}\label{eq:nonlinear_system}
\dot{x} = f(x)
\end{equation}
where $x \in M \subset \mathbb{R}^n$, $f:\mathbb{R}^n \rightarrow \mathbb{R}^n$ is analytic.  The manifold, $M$, is the state-space of the dynamical system.  
We introduce the concepts of a Koopman generator and its associated multi-variate Koopman semigroup, following the exposition of \cite{budivsic2012applied}. 


For continuous nonlinear systems, the Koopman semigroup is a semigroup, a set with an associative binary operation, $\mathcal{K}_{t\in \mathbb{R}}$ of linear but infinite dimensional operators, $\mathcal{K}_t$, that acts on a space of functions, $\Psi$, with elements $y: M \rightarrow \mathbb{R}^m$.
We assume that each $y\in\Psi$ is differentiable with a bounded derivative. Our function, $y$, is an observable because it is a function of the state, $x$.  We say $\mathcal{K}_t:\Psi \rightarrow \Psi$ is an operator for each $t\geq 0$.  
The Koopman operator applies the transformation, 
\begin{equation}\label{eq:koopman_def} \mathcal{K}_t \circ y(x_0) = y \circ \Phi_t(x_0),\end{equation}
where $\Phi_t(x)$ is the flow map of the dynamical system (\ref{eq:nonlinear_system}) evolved forward up to time $t$, given the initial state $x_0$. Instead of examining the evolution of the state, the Koopman semigroup allows us to study the forward evolution of functions of state, $y(x)$ \cite{williams2015data}. 

The action of the generator, $\mathcal{K_G}$, on the observables, $y$, for the Koopman semigroup is defined as 
\begin{equation}
\mathcal{K_G} \circ y  \triangleq \lim_{t\rightarrow 0} \frac{ \mathcal{K}_t \circ y - y }{t}.
\end{equation}

When $y$ and $t$ are fixed, we see from Eq. (\ref{eq:koopman_def}) that $\mathcal{K}_t$ is a function of the state, $x$. Similarly, $\mathcal{K_G}$ may be understood as a function of $x$.  The Koopman generator is a state-dependant derivative operator, see Section 7.6 of \cite{lasota1998chaos}. It satisfies
\be\label{eq:dyn_sys_in_data}
\frac{d}{dt}y(x)=\mathcal{K_G}(x)\circ y(x).
\ee In rare cases the Koopman generator will not have a dependence on $x$, see Eq. (24) and the surrounding section in  \cite{brunton2016koopman}.

\subsection{Learning Koopman operators from data}
In discrete-time, data-driven Koopman operator learning, we have $r$  pairs of measurements of observables, \[(y(x_i), y(f(x_i))),\mbox{ for }i=1,2,...,r.\] In continuous-time (CT), data-driven Koopman generator learning, our pairs are measurements of observables and their derivatives \be(y(x_i), d(y(x_i))/dt),\mbox{ for }i=1,2,...,r.\ee

We briefly remark that, there are techniques that a CT data-driven Koopman generator learning problem working with discrete measurements as opposed to measurements paired with their derivatives. For example, see \cite{mauroy2019koopman}.

Because we are working with numerical computation, these measurements (and their derivatives, if applicable) are finite dimensional, call the dimension $m<\infty$. We assume that our measurements are higher dimensional than the underlying dynamic system, $m\geq n$, and that $y(x)$ is injective. We assume all measurements are noise-free.

The Koopman generator (for CT) $\mathcal{K_G}(x)$ is unknown, however, we assume it exists and satisfies \be \frac{d(y(x_i))}{dt} = \mathcal{K_G}(x_i)\circ y(x_i)\mbox{ for }i=1,2,...,r.\ee

While we do not know the values of $x$ from our data we can write down our derivative function as an implicit function of $x$, this function is, $F:\R^m\rightarrow \R^m$. To describe $F$ we define, $y=h(x)$. The function $h$ is invertible as we assume that $y$ is injective.  Now we write $F$ as \be\label{eq:implicit} F(y) = \mathcal{K_G}(h^{-1}(y))\circ y.\ee 

The implicit function, $F$, simply is Eq. (\ref{eq:dyn_sys_in_data}) rewritten in terms of $y(x)$ instead of $x$. This is an important distinction as the true system state $x$ is not necessarily known. We are working from measurements (observables) of the system in Eq. (\ref{eq:nonlinear_system}). We label these measurements $y$.   Given such measurements, it is customary to choose a finite set of $N$ {\em dictionary functions}, $\psi(y):\R^m\rightarrow \R^N$, and a constant matrix, $K\in\R^{N\times N}$ to model $F$. 

Note that we choose $\psi$ without knowing the equation for $y(x)$. The function $y(x)$ is an observable function, as it is a measured function of the state of the system given in Eq. (\ref{eq:nonlinear_system}). Likewise, as $\psi$ is a function of $y(x)$, $\psi$ is a function of $x$ as well. Therefore, in a data-driven setting, dictionary functions of $y(x)$, $\psi(y)$, are also observables of some Koopman generator. Typically, we choose $N$ such that $N>m$, and so $\psi(y)$ is a ``lifting'' of our data (see the use of ``lifting'' in \cite{korda2018convergence}).

In numeric methods for Koopman modeling, we approximate $F$ with the matrix-vector product, $K\psi(y)$.  Effectively, in a data-driven setting, this amounts to projecting the action of an infinite dimensional Koopman operator as matrix multiplication on sampled data space. For example, when $N>m$, we can define a projection function $\mathtt{P}:\R^N\rightarrow \R^m$ that maps the matrix-vector product of the approximate Koopman operator $K$  multiplying the dictionary function $\psi(y)$  to the vector field $F(y)$ from Eq. (\ref{eq:implicit}),  such that
\be \mathtt{P}\circ K \circ \psi \circ y \approxeq  F \circ y .\ee 
This, in summary, yields the following finite approximation to the Koopman generator equation
\be\frac{dy(x)}{dt}=\mathcal{K_G}(x)y(x) \triangleq F(y(x)) \approxeq \mathtt{P}(K\psi(y)).\ee 

The matrix, $K$, is a linear operator, to ensure that it behaves like a true Koopman generator we choose it, in conjunction with $\psi$, to satisfy \be\label{eq:koop_approx_behavior}\frac{d\psi(y)}{dt}\approx K\psi(y).\ee  

We are tasked to learn a numerical, finite-dimensional approximation $K$ of the Koopman generator, $\mathcal{K_G}$ and the set of dictionary functions, $\psi$.  This matrix $K$ is the Koopman generator approximation that acts specifically on data-centered evaluations of dictionary functions $\psi(y)$ rather than the observable function $y(x).$

\subsection{Problem Statement}
Finding the Koopman generator $\mathcal{K_G}$ from data is  difficult, due to the lack of knowledge about the true Koopman generator and the parametric form of the measurements.  So, we instead solve the following optimization problem in its place.
\begin{equation}\label{eq:objective}
\min_{K, \psi \in \Psi} \sum_{i=1}^{r} \left\Vert  \frac{d\psi(y(x_i))}{dt}  - K \psi(y(x_i)) \right\Vert.
\end{equation}

This is an abstraction of the discrete-time problem statement. In practice, derivatives are difficult to numerically compute and measure.

In this optimization we need to select dictionary functions, as well as a real-valued matrix, $K$. This optimization problem is non-convex, since the form of $\psi(y)$ is unknown or parametrically undefined.  The model dimension, $N$, is a hyperparameter of Eq. (\ref{eq:objective}).

In EDMD the dictionary functions $\psi(y)$ are predefined, drawn from a class of nonlinear functions, and assumed to be known.  Under these assumptions, Eq. (\ref{eq:objective}) is a convex optimization problem with a closed-form solution. 

By contrast, in deepDMD, the dictionary functions and $K$ are learned simultaneously during iterative training.  This is, of course, a nonlinear, non-convex optimization problem, for which we employ variants of the stochastic gradient descent algorithm, such as adaptive gradient descent (AdaGrad) \cite{duchi2011adaptive} or adaptive momentum (ADAM) \cite{kingma2014adam}.  


\subsection{Finite Closure}\label{sec:error}

Previous work characterizes the closure and convergence of Koopman models as additional dictionary functions are appended to the model for the DMD and EDMD algorithms \cite{arbabi2017ergodic, korda2018convergence}. We explore closure as an inherent property of a dictionary. We begin by understanding the property of subspace invariance, which corresponds to finite exact closure. Closure, in this context, refers to how the action of the Koopman generator on the dictionary functions does not give any function outside of the dictionary functions' span.  This span is the subspace that we refer to.

Let $S$, be the span of our dictionary functions. The set $S$ is a subspace of the set of all analytic functions, including $y$, that map $x$ to $\R$.

\begin{definition}
We say that our dictionary, $\psi(y)$, satisfies Koopman \textit{subspace invariance} when $\mathcal{K_G}\cdot \psi(y)\in S$. A \textit{Koopman subspace invariant} dictionary is said to satisfy \textit{finite exact closure}.
\end{definition}

If some element in a dictionary does not satisfy Koopman subspace invariance, then there exists some element in the dictionary that, when acted on by the Koopman operator, cannot be represented as a linear combination of dictionary functions.  In the context of data-driven Koopman learning this means that when our dictionary contains $N$ functions, no $N$ by $N$ matrix captures the precise action of the Koopman generator.  

Finite exact closure (Koopman subspace invariance) is unlikely to be achieved, as (1) the Koopman generator may need to be infinite dimensional and (2) in the case that it does not, it is difficult to engineer models with exact closure even with an explicit knowledge of Eq. (\ref{eq:nonlinear_system}).  So, we also consider an approximate notion of closure relevant to building models from data. 
\begin{definition}\label{def:uniform} We say $\psi(y):\R^m \rightarrow \mathbb{R}^{N}$ achieves finite $\epsilon$-closure or finite closure with $O(\epsilon)$ error when there exists a $K\in \mathbb{R}^{N\times N}$ and an $\epsilon > 0$ such that 
\begin{equation}
\frac{d(\psi(y))}{dt} = K \psi(y) + \epsilon(y),
\end{equation} for the vector field $F$.

Our dictionary $\psi(y)$ achieves \textit{finite approximate closure} when, for the vector field, $F$, and every $y$, the function $\epsilon(y)$ is a bounded for every $K$ such that $||K||<\infty$.

We say that $\psi(y)$ achieves {\em uniform finite approximate closure} for some set $\mathcal{R}$ when it achieves finite closure with $|\epsilon(y)| < B \in \mathbb{R}$ for all $y \in \mathcal{R}.$
\end{definition}

We are most interested by the property of uniform finite approximate closure, especially when we can bound the constant, $B$, to be arbitrarily low. When $B$ can be bound in such a manner we have an accurate, data-driven model.

\begin{example}\label{example}
Unfortunately, closure does not come with every dictionary of observables, even if that dictionary spans the function space of the dynamic system. For example, a canonical polynomial basis, $\{1, y, y^2,...,y^n\}$, used to approximate the one-dimensional system $f(y)=y^2$, spans the dynamic system, but no model using this basis will be closed. In fact, none will achieve uniform finite approximate closure.  We illustrate what this lack of closure means when $y>>1$. In that case, the Lie derivative of $y^n$ is  \be\frac{d(y^n)}{dt} = \frac{d(y^n)}{dy}\frac{dy}{dt} = ny^{n-1}y^2  = ny^{n+1}.\ee Approximating this derivative when $y>>1$ as an $n^{th}$ degree polynomial will dramatically fail as the error will be of order $O(y^{n+1})$.  This failure will cascade back through approximations of all the other dictionary functions. Closure is a crucial property of these models! 
\end{example}

We want models with uniform finite approximate closure because, as the bounding constant goes to zero, $B \rightarrow 0$, we may use $K$ to perform stability, observability and spectral analysis. We see this is true as, in the discrete-time formulation, as $B\rightarrow 0$, $K$ approaches a projection of the action of the Koopman operator \cite{korda2018convergence}. 
When our dictionary is state-inclusive, it is trivial to project from $\psi$ to $y$ and its trajectory may yield stability insights. To keep our models meaningful in this way, all the dictionaries in this article are state-inclusive.

\section{Notation}\label{sec:AppendixNotation}
Our models, throughout this article, will have two classes of parameters. Each class refers to the distinct geometric properties of center and steepness. To create a clearer separation of function variables and parameters we will use the following notation \[\Eta(y;\mu_l,\alpha_l).\]  In this notation, $\Eta$ is a function whose variable is the vector $y$ and whose parameters are vectors $\mu_l$ and $\alpha_l$. The vectors $\mu_l$ and $\alpha_l$ refer to the geometrically distinct classes of parameters, $\mu$ for center and $\alpha$ for steepness.  A second example would be \[\eta(y_i;\mu_{li}, \alpha_{li}).\] In this notation, $\eta$ is a function whose variable is the scalar $y_i$, and whose parameters are the scalars $\mu_{li}$ and $\alpha_{li}$.  To make our equations less cumbersome we summarize all the distinct parameters with the label $\theta$. For example, \[\Eta(y;\mu_l,\alpha_l)\triangleq \Eta(y;\theta_l),\] and \[\eta(y_i;\mu_{li}, \alpha_{li})\triangleq \eta(x_i;\theta_{li}).\]

In general, our notation uses the following conventions. \begin{enumerate}
    \item Integer $i$ will be an index of measurement dimension.
    \item Integer $j$ will be an index of the first group of added dimensions.
    \item Integer $k$ will be an index of the second group of added dimensions.
    \item Integer $l$ will be an additional added dimension index.
    \item Integer $n$ will be the state dimension.
    \item Integer $m$ will be the number of  measurements.
    \item Integer $N_L$ will be the added dimensions. We subscript with $L$ in this paper as the added dimensions will correspond to conjunctive \textit{logistic} dictionary functions. 
    \item The $m\times N_L$ real-valued matrix, $w$, will be a matrix of weights. In our analysis it corresponds to a block of the Koopman Operator approximation matrix, $K$, which describes the flow of the state variables as a linear combination of nonlinear observables.
\end{enumerate}

\section{The SILL Dictionary: A Model of deepDMD's Learned Dictionary}\label{sec:SILL}
In this section, we define a new class of dictionary functions, $\psi(y)$, identified from deepDMD's solution to Eq. (\ref{eq:objective}). We will call this class \textit{State-Inclusive Logistic Liftings} (SILLs). We call them this because
\begin{enumerate}
    \item they contain the state of the vector field $F(y)$, the dynamics of the governing equations (as we have measured them) are directly included in the dictionary, so they are state-inclusive\footnote{Note that the functions that we define technically are measurement-lifting, not state-lifting. This is a consequence of how we cast our problem formulation.},
    \item the dictionaries in this class also contain nonlinear functions, all of which are conjunctive logistic functions, so these dictionaries are logistic in nature, and
    \item because each has at least one nonlinear dictionary function, the dictionary size, $N=1+m+N_L$, is greater than the number of measurements, $m$. Since $N>m$, the model using this dictionary is ``lifted'' to a higher dimension than the original measurements.
\end{enumerate}   

We previously showed that Koopman models chosen from SILL dictionaries have successfully learned global nonlinear phase-space behavior of several simple, nonlinear systems \cite{johnson2018class}.  In Section \ref{sec:SILLclosure}, we show, for the first time, that SILL dictionaries satisfy uniform finite approximate closure.


\subsection{The SILL Lifting Functions and deepDMD}
We define a multivariate conjunctive logistic function, $\Lambda:\R^m\rightarrow\R$, for a given center parameter vector $\mu_j \in \mathbb{R}^m$ and steepness parameter vector $\alpha_j \in \mathbb{R}^m$   as follows
\begin{equation}
\Lambda(y; \mu_j, \alpha_j) \triangleq \prod_{i=1}^{m}\lambda(y_i; \theta_{ji}),
\end{equation}
where the measurements are $y \in \mathbb{R}^m$ and the scalar logistic function, $\lambda:\R\rightarrow\R$, is defined as 
\begin{equation}\label{logistic}
\lambda(y_i; \mu_{ji}, \alpha_{ji}) \triangleq \frac{1}{1+e^{-\alpha_{ji} (y_i-\mu_{ji})}}.
\end{equation}
The parameters $\mu_{ji}$ define the centers or the point of activation for $\Lambda(x;\theta_j)$ along dimension $y_i$.  The parameter $\alpha_{ji}$ is a steepness or sensitivity parameter, and determines the steepness of the logistic curve in the $i^{th}$ dimension for $\Lambda(y;\theta_j)$. Conjunctive logistic functions map orthants of $R^m$ to be nearly 1 and the rest of the space to be nearly 0.

To illustrate how this conjunctive logistic function works, consider what happens if you set the vector $y$ to be constant in all but the $l^{th}$ dimension. Then \be\Lambda(y;\mu_j,\alpha_j) = (\prod_{i\neq l}c_i)\lambda(y_l;\mu_{jl}, \alpha_{jl}).\ee This is a constant times the logistic function in the $l^{th}$ dimension. When we project dictionary functions learned by deepDMD to a single dimension we observe that they likewise approximate a scaled logistic function.

Given $N$ multivariate logistic functions, we  define a SILL dictionary as $\psi : \R^m\rightarrow \R^{N}$, so that:
\begin{equation}
\psi(y) \triangleq \begin{bmatrix}
1&
y^T&
\bar\Lambda(y)^T
\end{bmatrix}^T
\end{equation}
where $\bar\Lambda(y) = [\Lambda(y;\theta_1), \hdots,\Lambda(y;\theta_N) ]^T$ is a vector of conjunctive logistic functions and $N=1+m+N_L$. We then have that $K\in\R^{N\times N}$.   This basis is measurement-inclusive. Ideally $\bar\Lambda(y)$, the measurements themselves and a constant spans each dimension of the vector field over the region of interest.


\begin{figure*}[ht]
    \centering
    \includegraphics[width=450pt]{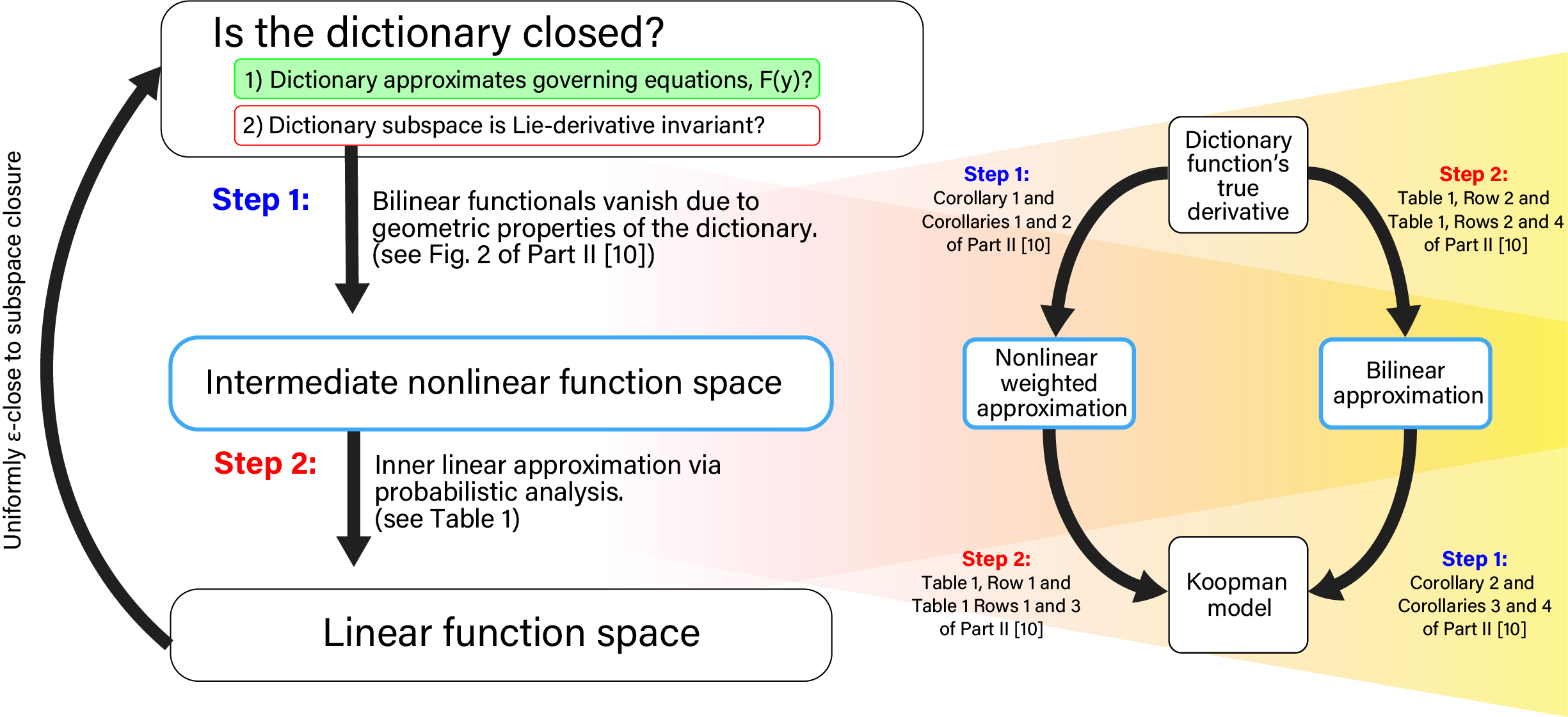}
    \caption{A visual outline of the (commutative) steps taken to prove uniform finite approximate closure of SILL dictionaries in Section \ref{sec:SILLclosure}. These same steps are applied in Part II of this paper to the heterogeneous augSILL dictionary, see \cite{johnson2022augSILL}.}
    \label{fig:paper_summary}
\end{figure*}

\section{Uniform Finite Approximate Closure of the SILL Dictionary: Step 1}\label{sec:SILLclosure}
We showed in \cite{johnson2018class} that a pure, measurement-inclusive, SILL dictionary can make effective low-dimensional Koopman models. We demonstrate that the SILL dictionary satisfies uniform finite approximate closure (see Definition \ref{def:uniform} in Section \ref{sec:error}).  In essence, uniform finite approximate closure guarantees that the dimensionality of the dictionary space does not need to diverge to infinity, while simultaneously approximating the vector field $F(y)$ sufficiently well.  This property is what makes numerical approximation of the Koopman generator equation possible. 

Recall from Example \ref{example} that when considering a new set of dictionary functions, we also have to consider the effects of dictionary explosion.  In Example \ref{example}, we computed the Lie derivatives of each dictionary function, which in turn generated new product terms that were not in the span of the existing dictionary.  Thus, for any new class of dictionary functions, a key property requisite to uniform approximate finite closure is the ability to approximate products of dictionary functions as elements of the span of the dictionary. 

In the context of the SILL dictionary, we show how approximate products of dictionary functions are elements in the span of our dictionary by showing convergence in steepness of \textit{products} of conjunctive logistic functions to a \textit{single} conjunctive logistic function. To show these bilinear terms are approximately in the span of our dictionary, we need to show that the product of two conjunctive logistic functions may be approximated as a single conjunctive logistic function as follows:
\begin{equation}\label{eq:approx1}
    \Lambda(y;\theta_l)\Lambda(y;\theta_j) \approx \Lambda(y;\theta^*)
\end{equation}
where $\theta^* = (\mu_{max}(l,j), \alpha_{max}(l,j))$ and 
\[
\mu_{max}(l,j) = \left( \max\{\mu_{l1},\mu_{j1}\}, ..., \max\{\mu_{lm},\mu_{jm}\} \right),
\] and $\alpha_{max}(l,j)$ are the $\alpha$ (steepness) values that correspond to the indices of the $\mu$'s.

Theorem \ref{thm:SILLconv} demonstrates that Eq. (\ref{eq:approx1}) is a good approximation in the limit of increasing steepness parameter $\alpha$. Specifically, it says that a bilinear combination of conjunctive logistic functions can be approximated by a single conjunctive logistic function.  This is important to show the uniform finite approximate closure of the SILL dictionary as it paves the way for Corollaries \ref{cor:logSILLApprox1} and \ref{cor:logSILLApprox2}. These corollaries are key steps in understanding the error between the Koopman model using the SILL dictionary and the true Lie derivatives of these dictionary functions (see \textbf{Step 1} in Fig. \ref{fig:paper_summary}).   The proofs of Theorem \ref{thm:SILLconv} and Corollary \ref{cor:logSILLApprox1} are in Section \ref{sec:AppendixProofs} of the Appendix.

To characterize the value of our approximation we need to define the regions where our approximation's efficacy is capped. Theorem \ref{thm:SILLconv} will demonstrate that the approximation in Eq. (\ref{eq:approx1}) can be arbitrarily good outside of a special set of hyperplanes.
Given a SILL dictionary, $\psi(y)$, we define a set of $mN_L$, $m-1$ dimensional hyperplanes in $\R^m$ corresponding to the centers of each 
$M_{\bar\Lambda} \triangleq \{ y:\mbox{there exists } i\in\{1,2,...,m\} \mbox{ and } j\in \{1,2,...,N_L\} \mbox{ so that } y_i=\mu_{ji} \}.$ This defines specific measurements where the approximation in Eq. (\ref{eq:approx1}) has maximal error and the error cannot be reduced. In practice, this does not stop the SILL dictionary from achieving high levels of performance for system identification \cite{johnson2022augSILL}. This makes sense as the set $M_{\bar\Lambda}$ is a finite collection of $m-1$ dimensional sets in $\R^m$, and is therefore of Lebesgue measure zero.

\begin{theorem}\label{thm:SILLconv}
Under Assumption \ref{assump:order}, if the measurements, $y$, do not exactly match the SILL dictionary functions corresponding center parameters, $y\not\in M_{\bar\Lambda}$, then, as the steepness parameters go to infinity, the product of two conjunctive logistic function will exponentially approach a single conjunctive logistic function in the dictionary, specifically, $\alpha\rightarrow\infty$, \[\Lambda(y;\theta_l)  \Lambda(y;\theta_j) - \Lambda(y;\theta^*)\rightarrow 0\] at a rate of $e^{-c\alpha}$, for some positive constant $c$.
\end{theorem}

Theorem \ref{thm:SILLconv} implies an intermediate result to demonstrating uniform finite approximate closure of the SILL basis in approximating a Koopman generator. This result is given in Corollary \ref{cor:logSILLApprox1} and it connects Theorem \ref{thm:SILLconv} to the Lie derivative of a SILL dictionary function.  

Since our Koopman models are linear in nature, each Lie derivative will need to be approximated as a linear combination, or weighted sum, of dictionary functions. However, the Lie derivative of each dictionary function is a nonlinear weighted sum of bilinear terms. Corollary \ref{cor:logSILLApprox1} replaces all of the bilinear terms in this nonlinear weighted sum with individual dictionary functions.  This gives us an intermediate approximation that is nearly the linear combination of dictionary functions that we need for a Koopman model.

\begin{corollary}\label{cor:logSILLApprox1}
Under the assumptions of Theorem \ref{thm:SILLconv}, when $F$ is spanned by a SILL dictionary, the Lie derivative of a conjunctive logistic function exponentially approaches a finite weighted sum of conjunctive logistic functions as the steepnesses of the functions goes to infinity. Specifically, \be\label{eq:nonlinearCombforSILL}\dot\Lambda(y;\theta_l)\rightarrow \sum_{i=1}^{n}\sum_{j=1}^{N_L}\alpha_{li}w_{ij}(1-\lambda(y_i;\theta_{li}))\Lambda(y;\theta^*)\ee exponentially as $\alpha\rightarrow\infty$.
\end{corollary}

The intermediate approximation to the Lie derivative of a conjunctive logistic function given in Corollary \ref{cor:logSILLApprox1} will be further approximated as a Koopman model in Section \ref{sec:averageErrorSILL}. This intermediate approximation is not a linear combination of dictionary functions like a Koopman model would be. However, the error between this intermediate approximation and the true Lie derivative is uniformly bounded by the intermediate bounding constant \be \bar B_1 = \max_{y\in M} \sum_{i=1}^{n}\sum_{j=1}^{N_L}\alpha_{li}w_{ij}(1-\lambda(y_i;\theta_{li}))\varepsilon_{\Lambda_l\Lambda_j}(y, i,j),\ee where $\varepsilon_{\Lambda_l\Lambda_j}(y,i,j) \triangleq \Lambda(y;\theta_l)  \Lambda(y;\theta_j) - \Lambda(y;\theta^*)$ and $M$ is the set of possible measurement values. When $M=\R^m$ the intermediate bound $\bar B_1$ is finite as it is a sum and product of finite-valued functions. When $M=\R^m - M_{\bar\Lambda}$ the bound decreases exponentially as $\alpha$ becomes increasingly positive because $\varepsilon_{\Lambda_l\Lambda_j}(y,i,j)$ does so as well. This error bound for the intermediate approximation will be used in Section \ref{sec:VC} to bound the overall error of a Koopman model using the SILL dictionary. 

We now show the second step in an alternate path to characterizing the finite closure properties of the SILL dictionary. This is to approximate the SILL observables' Lie derivative as a linear combination of SILL basis functions, \be\label{eq:SILL_finalApprox}\sum_{i=1}^{m}\sum_{j=1}^{N_L}\alpha_{li}w_{ij}\Lambda(y;\theta^*), \ee where $l\in\{1,2,...,N_L\}$. This step corresponds to the lower right arrow in the right side of Fig. \ref{fig:paper_summary} and is given in Corollary \ref{cor:logSILLApprox2}.

\begin{corollary}\label{cor:logSILLApprox2}
Under the assumptions of Theorem \ref{thm:SILLconv}, the error between
\be\label{eq:SILL_linearLie} 
\sum_{i=1}^{m}\sum_{j=1}^{N_L}\alpha_{li}w_{ij}\Lambda(y;\theta_l)\Lambda(y;\theta_j)
\ee and Eq. (\ref{eq:SILL_finalApprox}) goes to zero exponentially as $\alpha\rightarrow \infty$.
\end{corollary}

\begin{proof}
The proof follows from a direct application of Theorem \ref{thm:SILLconv}.
$\blacksquare$\end{proof}

The error described for the approximation in Corollary \ref{cor:logSILLApprox2} is   uniformly bounded by \be \tilde B_2 = \max_{y} \sum_{i=1}^{m}\sum_{j=1}^{N_L}\alpha_{li}w_{ij}\varepsilon_{\Lambda_l\Lambda_j}(y,i,j), \ee where $\varepsilon_{\Lambda_l\Lambda_j}$ is defined as above, and this bound likewise goes to zero exponentially as $\alpha$ increases. The bounding constant $B$ which we will use to demonstrate uniform finite closure will be a function of $\bar B_1$, $\tilde B_2$ as well as two other bounds. This relationship is in Section \ref{sec:VC}.

The end result of Eq. (\ref{eq:nonlinearCombforSILL}) is a nonlinear combination of dictionary functions and Eq. (\ref{eq:SILL_linearLie}) is an approximation of the Lie derivative, $\dot\Lambda$.  We explore the approximation quality in Section \ref{sec:averageErrorSILL}. Section \ref{sec:averageErrorSILL} also ties these results to conclude that the SILL dictionary satisfies uniform finite approximate closure and that the bounding constant goes to zero, $B\rightarrow 0$, exponentially as steepness of the dictionary functions increases and the number of measurements increase. Thus, SILL dictionary functions define a spanning set for Koopman observables.

\begin{table*}[ht]
    \centering
    \begin{tabular}{|p{16mm}|p{44mm}|p{54mm}|p{28.5mm}|}
    \hline
       \textbf{Reference}  &  \textbf{Approximation} & \textbf{Difference (Error)} & \textbf{Error Bound} \\
    \hline
       (\ref{eq:nonlinearCombforSILL})   & \(\begin{aligned} &\sum_{i=1}^{m}\sum_{j=1}^{N_L} \alpha_{li}w_{ij} \Lambda(y;\theta^*) \end{aligned}\) & \(\begin{aligned} &\sum_{i=1}^{m}\sum_{j=1}^{N_L} \alpha_{li}w_{ij} \lambda(y_i;\theta_{li}) \Lambda(y;\theta^*) \end{aligned}\)  & \(\begin{aligned}&\sum_{i=1}^{m}\sum_{j=1}^{N_L}\frac{\nu_{ij}}{2^{m+1}} \end{aligned}\) \\
    \hline
        (\ref{eq:SILL_lie_derivative})  & \(\begin{aligned}  &\sum_{i=1}^{m}\sum_{j=1}^{N_L} \alpha_{li}w_{ij} \Lambda(y;\theta_l)\Lambda(y;\theta_j) 
    \end{aligned}\)  & \(\begin{aligned} 
&\sum_{i=1}^{m}\sum_{j=1}^{N_L} \alpha_{li}w_{ij} \lambda(y_i;\theta_{li})\Lambda(y;\theta_l)\Lambda(y;\theta_j) \end{aligned}\)  & \(\begin{aligned}&\sum_{i=1}^{m}\sum_{j=1}^{N_L}\frac{\nu_{ij}}{2^{2m+1}} \end{aligned}\) \\
    
    \hline
    \end{tabular}
    \caption{Approximations to and properties of error bounds for the equations referred to in the \textbf{Reference} column. The reference equation is approximated as the corresponding equation in the \textbf{Approximation} column. We give the error of this approximation in the \textbf{Difference (Error)} column. The \textbf{Error Bound} column gives a bound on this error.}
    \label{tab:linearityErrorAugSILL}
\end{table*}

\section{Uniform Finite Approximate Closure of the SILL Dictionary: Step 2}\label{sec:averageErrorSILL}

This section simultaneously addresses the approximation of two related mathematical objects.
\begin{enumerate}
    \item The nonlinear combination of SILL dictionaries in Equation (\ref{eq:nonlinearCombforSILL}) with a linear combination. 
    \item The bilinear combination in Equation (\ref{eq:SILL_linearLie}) with the Lie derivatives of a conjunctive logistic function.
\end{enumerate}
So, there are two distinct approximations, one for Corollary \ref{cor:logSILLApprox1} and one for Corollary \ref{cor:logSILLApprox2}. The explicit approximations are given in Table \ref{tab:linearityErrorAugSILL}. In the sense of the steps shown in Fig. \ref{fig:paper_summary}, each of these approximations is a step closer to the Koopman model. The relationship between these approximations, Corollaries \ref{cor:logSILLApprox1} and \ref{cor:logSILLApprox2}, and the uniform finite approximate closure of the SILL dictionary is schematically depicted in Fig. \ref{fig:paper_summary}.

To show uniform finite approximate closure, we need to characterize the error, \be\label{eq:overallError} \epsilon(y) = \frac{d\psi(y)}{dt} - K\psi(y).\ee  We do so by choosing a specific approximation for the time derivative of each nonlinear dictionary function in the SILL basis.

Characterizing the closure of a dictionary means understanding the error between the Koopman model built with that dictionary and the true Lie derivative of each dictionary function. In Section \ref{sec:SILLclosure} we used mathematical analysis to show that the approximation error of some models will go to zero in the limit of high steepness. These models do not fully bridge the gap between our Koopman model and the true Lie derivative. This section characterizes additional approximation errors to fully connect our Koopman model to the true Lie derivative. Then, we uniformly bound those errors to show the uniform finite approximate closure of the SILL dictionary.

For all $l\in\{1,2,...,N_L\}$, the error of approximating Eq. (\ref{eq:SILL_nonlinComb}) with Eq. (\ref{eq:SILL_finalApprox}) is: 
\be\label{eq:SILL_lin_error}
 \sum_{i=1}^{m}\sum_{j=1}^{N_L}\alpha_{li}w_{ij}\lambda(y_i;\theta_{li})\Lambda(y;\theta^*), 
\ee  and the error of approximating Eq. (\ref{eq:SILL_lie_derivative}) with Eq. (\ref{eq:SILL_linearLie}) is: \be\label{eq:SILL_linProd_error}
 \sum_{i=1}^{m}\sum_{j=1}^{N_L}\alpha_{li}w_{ij}\lambda(y_i;\theta_{li})\Lambda(y;\theta_l)\Lambda(y;\theta_j).
\ee

To get a grasp on the kind of errors we can expect when modeling with the SILL basis we need to get a grasp on the magnitude of Eq. (\ref{eq:SILL_lin_error}) and Eq. (\ref{eq:SILL_linProd_error}).

We do so by sampling possible parameter and state values from uniform distributions over a symmetric interval and then computing distributions of one dimensional logistic functions.  We then do the same for terms in the sums in Equations (\ref{eq:SILL_lin_error}) and  (\ref{eq:SILL_linProd_error}), under the assumption that $w_{ij}=1$. In doing so, we notice an exponential decrease in expected error as the number of measurements increases.

\subsection{Expectation of Approximation Error Vanishes}\label{sec:augSILL_error}

To understand the exponential decrease in expected approximation error we start with the expected values of a single dimensional logistic function. We do so with parameters and measurement values sampled from uniform distributions defined on the interval $[-a,a]$.  We choose this statistical model for how our data and parameters are sampled, because  1) the data and parameters are assumed to belong to a bounded continuum, and 2) the uniform distribution is the maximum entropy distribution for a continuous random variable on a finite interval.  Since our error terms are weighted sums of products of  logistic functions we, under the assumption of independence, can compute the expected value of our error terms as a weighted sum of the distinct expected values of logistic functions. This is justified by the linearity of and product rule of expectation under independence. 

We cannot explicitly compute the probability density function (PDF) of our logistic functions, so, we compute the values of these integrals numerically. Intermediate steps and details of this approximation are in Section \ref{sec:AppendixPDF} of the Appendix. In Fig. \ref{fig:expectedVals} we show their calculated expected values and variances for symmetric uniform distributions with different values of $a$. 

We find that the expected value of a logistic function will be $1/2$ regardless of the sampling interval (see Fig. \ref{fig:expectedVals}). Its variance, as we sample in a wider interval, tends to the functional extremes of zero and one.  This is favorable for the linearity of our approximation since, for all $\varepsilon\in (0, 0.5]$, $(0.5-\varepsilon)(0.5+\varepsilon) = 0.25 - \varepsilon^2 < 0.25 = (0.5)^2$. So, products of more extreme samples are lower in value than products of samples near the expected value. This demonstrates that our bound, and the error bounds that follow, computed using the expected value of $1/2$, are conservative.

\begin{figure}[ht]
    \includegraphics[width=\linewidth ]{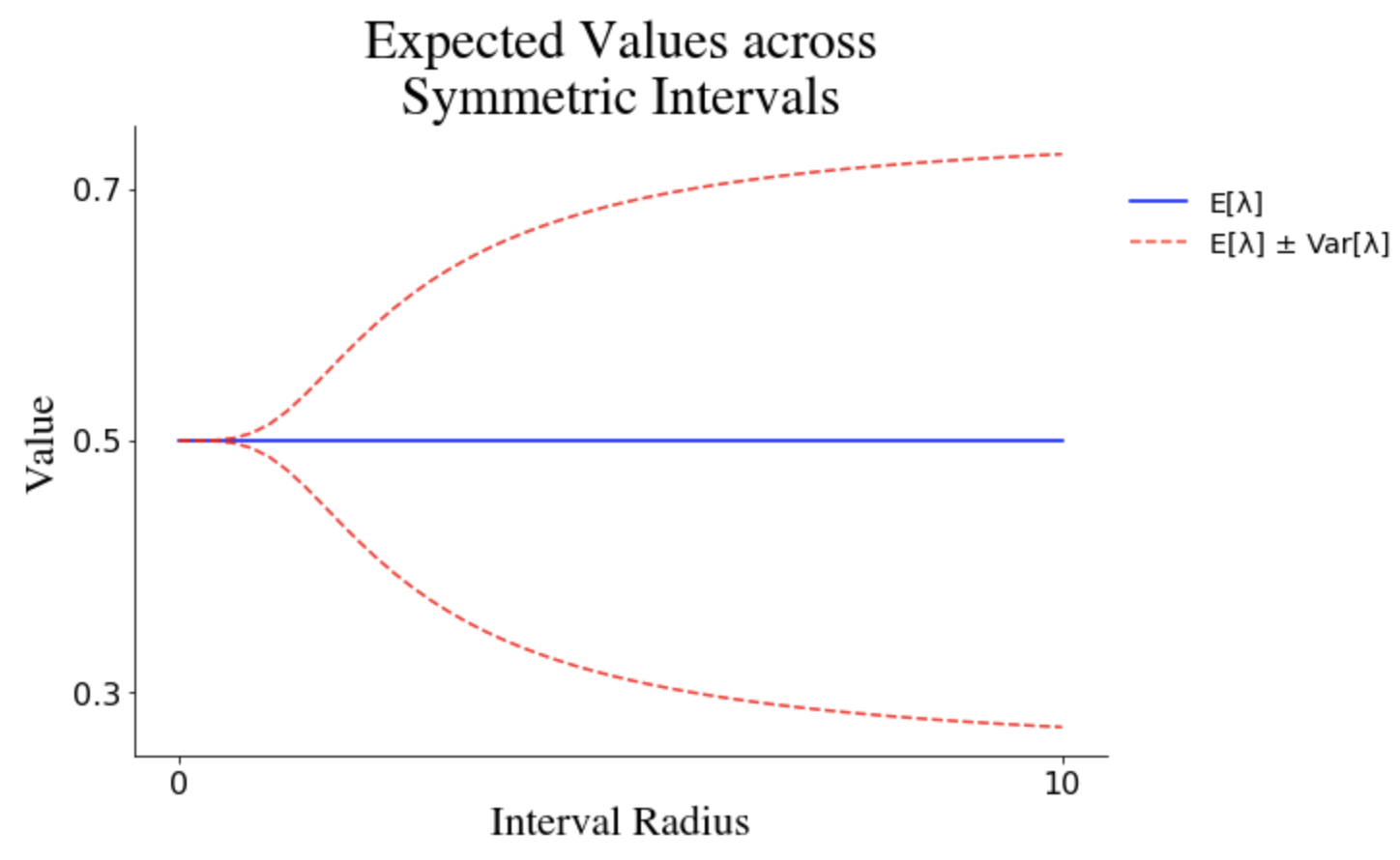}
    \caption{Expected values and variances of a logistic function with parameters and measurement values sampled from symmetric uniform distributions of various interval radii. Note that the expected value is always $1/2$.}
    \label{fig:expectedVals}
\end{figure}


We approximate a single term in the sum of the error function and extrapolate via the sum and product rule of expectation under the assumption of independence to see how nearly linear our approximation is (see Section \ref{sec:AppendixPDF} of the Appendix). We can conservatively bound the expectation of the approximation error as a product that decreases exponentially with the number of measurements.  The expected value of a conjunctive logistic function is \be\label{eq:CEBlogRBF} E[\Lambda]<\frac{1}{2^m}.\ee

We now compute the expected error given the expected value of the logistic and conjunctive logistic functions. The end result is that, applying the independence assumption, one can conservatively expect the error of each summation term in Eq. (\ref{eq:SILL_lin_error}) to decrease in the numbers of measurements, $m$, at a rate bound above by $\frac{1}{2^{m+1}}$ and  the error of each summation term in Eq. (\ref{eq:SILL_linProd_error}) to decrease in $m$ at a rate bound above by $\frac{1}{2^{2m+1}}$. In each case, these error bounds are exponentially decaying in $m$, the dimension of the measurements.

Explicitly, these intermediate error bounds are \be\label{eq:SILLerrorBounds} \bar B_2 = \sum_{i=1}^{m}\sum_{j=1}^{N_L}\frac{\nu_{ij}}{2^{m+1}} \mbox{ and }\tilde B_1=\sum_{i=1}^{m}\sum_{j=1}^{N_L}\frac{\nu_{ij}}{2^{2m+1}}, \ee where $\nu_{ij}\in\R$, where each $\nu_{ij}<a^2$ as each $\alpha_{li}$ and $w_{ij}$ is sampled from the interval $[-a, a]$.  A summary is given in Table \ref{tab:linearityErrorAugSILL}.

\subsection{Error of the final model}\label{sec:VC}
Theorem \ref{thm:SILLconv} shows that in the limit of infinitely steep dictionary functions, the error of approximating the product of two conjunctive logistic functions with a single one goes to zero exponentially. This gave us two approximation corollaries that share exponentially decreasing error.  Likewise, our probabilistic bound demonstrated that the error in approximating the Lie derivative of SILL functions as a Koopman model goes to zero exponentially, on average, as the number of measurements increases.  This means that the error, $\epsilon_l(y)$, of the approximation 
\be \label{eq:SILLapproximation}
\dot{\Lambda}(y;\theta_l)& =
\sum_{i=1}^{m}\sum_{j=1}^{N_L}\alpha_{li}w_{ij}(1-\lambda(y_i;\theta_{li}))\Lambda(y;\theta_l)\Lambda(y;\theta_j)\\& \approx \sum_{i=1}^{m}\sum_{j=1}^{N_L}\alpha_{li}w_{ij}\Lambda(y;\theta^*)
\ee will be bound by the sum $\bar B_1 + \bar B_2$, as well as the sum $\tilde B_1 + \tilde B_2$\footnote{If $|\dot\Lambda-A_1|<\bar B_1$ and $|A_1 - A_2|<\bar B_2$, where $A_1$ is the limiting approximation in Eq. \ref{eq:nonlinearCombforSILL} and $A_2$ is the approximation on row 1 of Table \ref{tab:linearityErrorAugSILL}. Then $|\dot\Lambda-A_2|< \bar B_1 + \bar B_2$ by the triangle inequality. This argument applies to $\tilde B_1$ and $\tilde B_2$ without loss of generality where $A_1$ is the approximation on row 2 of Table \ref{tab:linearityErrorAugSILL} and $A_2$ is Eq. (\ref{eq:SILL_finalApprox}).}. This means that the error will   be uniformly bounded by the constant \be\min\{\bar B_1 + \bar B_2, \tilde B_1 + \tilde B_2\} \triangleq B>0.\ee The error will also exponentially go to zero with increasingly steep logistic functions and an increasingly large number of measurements.  The error bounds tied to the number of measurements can be used to choose the number of observables to measure to build a good Koopman model.

The bound, $B$, on our closure error, $\epsilon(y)$, goes to zero, $B\rightarrow 0$ at the extremes of large measurements and steep dictionary functions. This bound shows that the SILL dictionary satisfies uniform finite approximate closure. 

\be
|\dot\Lambda_l - A_2| < B(\alpha, m)\leq \bar B_1(\alpha, m) + \bar B_2(\alpha, m) \rightarrow 0
\ee exponentially as $\alpha$ and $m$ increase when $y\in\R^m-M_{\bar\Lambda}$ and $|\theta|$, $|K|$ are bounded. These are reasonable assumptions as the set $M_{\bar\Lambda}$ is of measure zero in $\R^m$ and that $|\theta|$, $|K|$ being bounded is required for any numerical application. Note that $B(\alpha, m)$ is constant given a set of possible measurements and a bound on the model parameters. It only varies in model parameters and hyperparameters.

In practice, these limiting conditions of extreme steepness and measurement dimension are not required. This constant in measurements, $B$, appears to be a wide upper bound. One captures strong nonlinear system features using the SILL basis for a low dimensional dictionary and few measurements with no restrictions on the steepness parameter  \cite{johnson2018class}.

\section{Conclusion}\label{sec:conclusion}
Learning models in a data-driven setting using human-defined dictionaries results in high-dimensional, over-parameterized representations of what could be simple physical phenomena. DeepDMD and other ANN learning techniques address these issues but at the cost of extensive computational time. Moreover, the dictionaries learned by deepDMD are ad hoc, randomly constructed, and difficult to interpret. 

To investigate how deepDMD finds successful Koopman invariant subspaces, we extracted simple features of the dictionary functions learned by deepDMD and looked for properties to explain their success.  We discovered a Koopman dictionary which we call the State-Inclusive Logistic 
Lifting (SILL) dictionary. We showed, in Section \ref{sec:SILLclosure}, that SILL model error drops exponentially when steepness parameters increase and more measurements are included in the data. This suggests that SILL dictionaries can be used to construct high-fidelity, low-dimensional, structured, and interpretable Koopman models.  Quantifying how many measurements should be taken to build an accurate Koopman model is an intriguing future research direction.

The success of any measurement-inclusive Koopman dictionary depends on more than how the dictionary approximates the evolution of the measurements in time. It also depends on how well linear combinations of the dictionary functions approximate the dictionary function's Lie derivatives. The SILL dictionary is fully specified with closed-form analytical expressions (unlike deepDMD) and as a consequence of the theoretical results in this paper, satisfies a unique numerical property of uniform approximate finite closure.  Our methodology provides a template for understanding how deep neural networks successfully approximate governing equations \cite{brunton2016discovering}, the action of operators and their spectra \cite{mezic2005spectral}, and dynamical systems \cite{mezic2013analysis}.  Further, these results provide a pattern for improving scalability and interpretability of dictionary-based learning models for dynamical system identification.

More can be done to improve dimensional scalability of dictionary models. The deepDMD algorithm suggests further innovations to improve dictionary learning. One such innovation that we see in the dictionaries learned by deepDMD is the use of heterogeneous dictionaries. This fruitful research direction will be given in part 2 of this paper \cite{johnson2022augSILL}.


\begin{ack}
Any opinions, findings and conclusions or recommendations expressed in this material are
those of the author(s) and do not necessarily reflect the
views of the Defense Advanced Research Projects Agency
(DARPA), the Department of Defense, or the United States
Government. This work was supported partially by a Defense Advanced Research Projects Agency (DARPA) Grant
No. FA8750-19-2-0502, PNNL Grant No. 528678, ICB Grant No. W911NF-19-D-0001 and No. W911NF-19-F-0037, and ARO Young Investigator Grant No. W911NF-20-1-0165.
\end{ack}

\bibliographystyle{plain}        
\bibliography{bibliography}           

\appendix

\section{Proofs of Theorems}\label{sec:AppendixProofs}

In \cite{johnson2018class} we demonstrated that imposing a total order on conjunctive logistic basis functions implied some basic closure properties of the SILL dictionary.  We show in this section proofs that apply this total order as steps to showing the uniform finite closure of the SILL and augSILL dictionaries. The needed total order is summed up in the following assumption.
\begin{assumption}\label{assump:order}
There exists a total order on the set of conjunctive logistic functions, $\Lambda(y;\theta_1) ,..., \Lambda(y;\theta_l)$, induced by the positive orthant $\mathbb{R}^m_+$, where $\theta_{l} \gtrsim \theta_{j}$ whenever $\mu_j - \mu_l \in \mathbb{R}^m_+$, and therefore $\Lambda(y;\theta_l) \geq \Lambda(y;\theta_j)$. 
\end{assumption}  

\begin{remark}
Given a finite SILL dictionary that does not satisfy Assumption \ref{assump:order}, one can enforce that Assumption \ref{assump:order} holds by adding a finite number of additional conjunctive logistic functions to the basis.
\end{remark}

In these results we do not consider the error of approximating the vector field, $F$, we instead focus on the subspace invariance of a model built from a dictionary that already captures the basic system dynamics. This assumption is given formally below. 

\begin{assumption}
The vector field, $F$, lies in the span of our SILL (or alternativly, augSILL) dictionary.
\end{assumption}

\textit{Theorem} \ref{thm:SILLconv}:
Under Assumption \ref{assump:order}, if the dictionary functions do not exactly match their corresponding center parameters, $y_i\neq\mu_{ji} $ for all $ i\in \{1, 2, ..., m\}$ and  $j\in\{1,2,...,N\}$, then, as the steepness parameters go to infinity, the product of two conjunctive logistic function will exponentially approach a single conjunctive logistic function in the dictionary, $\alpha\rightarrow\infty$, \[\Lambda(x;\theta_l)  \Lambda(x;\theta_j) - \Lambda(x;\theta^*)\rightarrow 0\] exponentially for any $l\in\{1,2,...,N_L\}$.

\begin{proof}
We now investigate the error of the approximation. Eq. (\ref{eq:approx1})'s error of approximation is 
\begin{equation}\label{Lambda_errorterm}
\begin{aligned}
\Lambda(y&;\theta_l)  \Lambda(y;\theta_j) - \Lambda(y;\theta_l) \\ 
&=\frac{1-\Lambda(y;\theta_j)^{-1}}{(\Lambda(y;\theta_l)\Lambda(y;\theta_j))^{-1}} \\
&= \frac{1 -  (1+e^{-\alpha_{l1}(y_1 - \mu_{l1})})   ... (1+e^{-\alpha_{lm}(y_m - \mu_{lm})}) }{\prod_{i=1}^{m}(1+e^{-\alpha_{li}(y_i - \mu_{li})}) (1+e^{-\alpha_{ji}(y_i - \mu_{ji})})},
\end{aligned}
\end{equation} whenever $\theta_l\gtrsim\theta_j.$

Without loss of generality, assume that $\theta_l\gtrsim\theta_j$. As we hold $y$ constant and let $\alpha\rightarrow \infty$,  for all $ i\in \{1, 2, ..., m\}$, and any possible value of $l,j$ we observe two cases. In these cases we denote a possible value of $\mu_{li}$, $\mu_{ji}$ etc. as $\mu_*$

\textbf{Case 1}, $y_i-\mu_* > 0$: As $\alpha \rightarrow \infty$ we have that $e^{-\alpha(y_i-\mu_*)}\rightarrow 0$ and so $\frac{1}{1+ e^{-\alpha_i(y_i-\mu_*)}} \rightarrow 1$. 

\textbf{Case 2}, $y_i-\mu_* < 0$: As $\alpha \rightarrow \infty$ we have that $e^{-\alpha_i(y_i-\mu_*)}\rightarrow \infty$ and so $\frac{1}{1+ e^{-\alpha_i(y_i-\mu_*)}} \rightarrow 0$ exponentially.  

So, if there exists $i\in \{1,2,...,m\}$ for every $ j$,  so that $y_i-\mu_{ji} < 0$, then Eq. (\ref{eq:approx1}) goes to 0 exponentially as $\alpha \rightarrow \infty$.  

However, if $y_i-\mu_{ji} > 0$ for all $i$, then, since $\theta_l \gtrsim \theta_j$, for all $i$ we have that $y_i-\mu_{li} > 0$ for all $i$ as well, so Eq. (\ref{Lambda_errorterm}) goes to $\frac{1-1}{1}=0$ exponentially as $\alpha \rightarrow \infty$.  
$\blacksquare$\end{proof}

\textit{Corollary}
\ref{cor:logSILLApprox1}:
Under the assumptions of Theorem \ref{thm:SILLconv}, and assuming that $F$ is spanned by a SILL dictionary, the Lie derivative of a conjunctive logistic function exponentially approaches a finite weighted sum of conjunctive logistic functions as the steepnesses of the functions goes to infinity. Specifically, \be\dot\Lambda(x;\theta_l)\rightarrow \sum_{i=1}^{n}\sum_{j=1}^{N_L}\alpha_{li}w_{ij}(1-\lambda(x_i;\theta_{li}))\Lambda(x;\theta^*)\ee exponentially as $\alpha\rightarrow\infty$.

\begin{proof} 
We assume that $F$ is spanned by our set of nonlinear dictionary functions.

This means that there exists a real-valued weighting matrix, $w\in \R^{m\times N_L}$, so that for any, $i\in\{1,2,...,m\}$, the $i^{th}$ element of $F$, $F_i$, can be written as:
\begin{equation}\label{eq:f_regression}
F_i(y) = \sum_{j=1}^{N_L} w_{ij}  \Lambda(y, \theta_j).
\end{equation}

Therefore, the time derivative of an arbitrary nonlinear observable in the SILL dictionary is: 

\begin{equation}\label{eq:SILL_lie_derivative} 
    \begin{aligned}
    \dot \Lambda(y;&\theta_l) = (\nabla_y\Lambda(y;\theta_l))^T\frac{dy}{dt} = (\nabla_y\Lambda(y;\theta_l))^TF(y)\\
    &=\sum_{i=1}^m\alpha_{li}(1-\lambda(y_i;\theta_{li}))\Lambda(y;\theta_l)F_i(y)\\
    &=\sum_{i=1}^m\alpha_{li}(1-\lambda(y_i;\theta_{li}))\Lambda(y;\theta_l)\sum_{j=1}^{N_L} w_{ij}  \Lambda(y; \theta_j)\\
    &=\sum_{i=1}^{m}\sum_{j=1}^{N_L}\alpha_{li}w_{ij}(1-\lambda(y_i;\theta_{li}))\Lambda(y;\theta_l)\Lambda(y;\theta_j).
    \end{aligned}
\end{equation}

Theorem \ref{thm:SILLconv} shows, that in the limit of steepness, the error between Eq. (\ref{eq:SILL_lie_derivative}) and \be\label{eq:SILL_nonlinComb}\sum_{i=1}^{m}\sum_{j=1}^{N_L}\alpha_{li}w_{ij}(1-\lambda(y_i;\theta_{li}))\Lambda(y;\theta^*) \ee will go to zero exponentially. $\blacksquare$\end{proof}

 \section{PDF of X(Y-Z) and Logistic Functions}\label{sec:AppendixPDF}
 Given the symmetric uniform distributed random variables X, Y and Z. We can compute the PDF of X(Y-Z), a term in logistic functions. The random variables X and Z represent the scalar parameters and the random variable Y represents the scalar measurement. So, X(Y-Z) corresponds to $\alpha(y-\mu)$. We choose X, Y and Z to be identically and independently distributed (iid) as the symmetric uniform distribution: $U(-a, a), a\in\R^+$.  We then apply the law of the unconscious statistician (LOTUS) to compute the expected value of a logistic function.

We compute the PDF of X(Y-Z), where the random variables X, Y and Z are independently and identically distributed as the symmetric uniform distribution: $U(-a, a), a\in\R^+$.

We start by computing the PDF of (Y-Z).  The PDF of each random variable is \be\label{eq:uniformPDF}
f_U(x) = \begin{cases}
 \frac{1}{2a} &\mbox{if }x\in [-a, a] \\
 0 &\mbox{otherwise}.
\end{cases}\ee  Since the random variables are symmetrically distributed, this is the same distribution as (Y+Z),  a symmetric triangular distribution.  The PDF is  \be\label{eq:PDFSum}
f_T(x) = \begin{cases}
 \frac{1}{2a} + \frac{x}{4a^2} &\mbox{if } x\in[-2a, 0) \\
 \frac{1}{2a} - \frac{x}{4a^2} &\mbox{if } x\in(0, 2a] \\
 0 &\mbox{otherwise}.
\end{cases}\ee 

Now we use the formula for the distribution of a product of random variables, \be g(z) = \int_{-\infty}^{\infty}f_T(x)f_U(z/x)\frac{1}{|x|}dx, \ee to compute the PDF of X(Y-Z).  So, \be g(z)&= \int_{-\infty}^{\infty}\begin{cases}
 \frac{1}{2a} + \frac{x}{4a^2} &\mbox{if } x\in[-2a, 0) \\
 \frac{1}{2a} - \frac{x}{4a^2} &\mbox{if } x\in(0, 2a] \\
 0 &\mbox{otherwise}
\end{cases}\\&\:\:\:\:\times\begin{cases}
 \frac{1}{2a|x|} &\mbox{if }\frac{z}{x}\in [-a, a] \\
 0 &\mbox{otherwise}
\end{cases}dx\\
&= \int_{0}^{2a}
 (\frac{1}{2a} - \frac{x}{4a^2} 
)\begin{cases}
 \frac{1}{2ax} &\mbox{if }\frac{z}{x}\in [-a, a] \\
 0 &\mbox{otherwise}
\end{cases} \\
&\:\:\:\:- \int_{-2a}^{0}
 (\frac{1}{2a} + \frac{x}{4a^2} 
)\begin{cases}
 \frac{1}{2ax} &\mbox{if }\frac{z}{x}\in [-a, a] \\
 0 &\mbox{otherwise.}
\end{cases}
\ee

The following equivalences hold, however they are only equal to $g(z)$ when $\frac{|z|}{a}\leq 2a$ or, alternatively written, $z\in [-2a^2, 2a^2]$.

\be
g(z) &= \frac{1}{4a^2}\int_{\frac{|z|}{a}}^{2a}
 (\frac{1}{x} - \frac{1}{2a} 
) - \frac{1}{4a^2}\int_{-2a}^{\frac{-|z|}{a}}
 (\frac{1}{x} + \frac{1}{2a} 
)\\
&= \frac{1}{4a^2}(\ln|x|-\frac{x}{2a})|_{x=\frac{|z|}{a}}^{x=2a} \\&- \frac{1}{4a^2}(\ln|x|+\frac{x}{2a})|_{x=-2a}^{x=\frac{-|z|}{a}}\\
&= \frac{1}{4a^2}(\ln(2a) - \ln(\frac{|z|}{a})-1+\frac{|z|}{2a^2}) \\&- \frac{1}{4a^2}(\ln(\frac{|z|}{a}) - \ln(2a)-\frac{|z|}{2a^2}+1)\\
&= \frac{1}{2a^2}(\ln\left(\frac{2a^2}{|z|}\right) + \frac{|z|}{2a^2} - 1).
\ee

So, we have that our final PDF is 
\be
g(z)= \begin{cases}
 \frac{1}{2a^2}(\ln(\frac{2a^2}{|z|}) + \frac{|z|}{2a^2} - 1) & \mbox{if } z\in I\\
 0 & \mbox{else},
\end{cases}
\ee where $I$ is the interval $[-2a^2, 2a^2].$

With this PDF we use the LOTUS to compute the expected value of a randomly sampled logistic function as \be 
E[\lambda] = \int_{-\infty}^{\infty}g(x)\left(\frac{1}{1+e^{-x}}\right)dx.
\ee

We compute the variance using the formula: $Var[X] = E[X^2] - E[X]^2$, which involves an additional application of the LOTUS.

To the authors' knowledge, none of these integrals has a closed form solution in terms of the sampling interval parameter, $a$. Hence, we use numeric simulation to understand this relationship (see Fig. \ref{fig:expectedVals}).

\end{document}